\documentclass[conference, onecolumn,12pt]{IEEEtran}
\usepackage{setspace}
\usepackage[margin=0.99in]{geometry}

\usepackage{amsthm}
\usepackage{amsmath, amssymb, amsfonts}
\usepackage{graphicx}
\usepackage{epsfig}
\usepackage{epstopdf}
\usepackage{verbatim}
\usepackage{hyperref}
\usepackage{color}
\usepackage{caption}
\usepackage{subcaption}
\usepackage{cite}
\usepackage[none]{hyphenat}

\newcommand{\diff}{{\rm d}}
\newcommand{\Exp}{\mathsf E}

\newcommand{\Prob}{{\mathsf P}}
\newtheorem{theorem}{Theorem}% \newtheorem{theorem}{Theorem}[section]
\newtheorem{lemma}{Lemma} %\newtheorem{theorem}{Lemma}

\newtheorem{profile}{Profile}

\IEEEoverridecommandlockouts

\begin{document}
\sloppy % you'll see far less words spill into the margin
%%%%%%%%% TITLE
\title{ Sharing of Unlicensed Spectrum \\ by Strategic Operators \thanks{This work was supported in part by a gift from Futurewei Technologies
and the National Science Foundation under Grant Nos. ECCS-1231828 and CCF-1423040.}}
\author{
 Fei Teng, Dongning Guo and Michael L. Honig \\
Northwestern University, 2145 Sheridan road, Evanston, IL 60201
}
\author{\IEEEauthorblockN{Fei Teng, Dongning Guo, and Michael L. Honig}
\IEEEauthorblockA{Department of Electrical Engineering and Computer Science, Northwestern University\\ Evanston, IL 60208, USA
% \\  \\ \Large{March 26, 2015}
}}

\maketitle
% add page number
\thispagestyle{plain}
\pagestyle{plain}

    \begin{abstract}

    Facing the challenge of meeting ever-increasing demand for wireless data, the industry is striving to exploit large swaths of spectrum which anyone can use for free without having to obtain a license. Major standards bodies are currently considering a proposal to retool and deploy Long Term Evolution (LTE) technologies in unlicensed bands below 6 GHz. This paper studies the fundamental questions of whether and how the unlicensed spectrum can be shared by intrinsically strategic operators without suffering from the tragedy of the commons. A class of general utility functions is considered. The spectrum sharing problem is formulated as a repeated game over a sequence of time slots. It is first shown that a simple static sharing scheme allows a given set of operators to reach a subgame perfect Nash equilibrium for mutually beneficial sharing. The question of how many operators will choose to enter the market is also addressed by studying an entry game. A sharing scheme which allows dynamic spectrum borrowing and lending between operators is then proposed to address time-varying traffic and proved to achieve perfect Bayesian equilibrium. Numerical results show that the proposed dynamic sharing scheme outperforms static sharing, which in turn achieves much higher revenue than uncoordinated full-spectrum sharing. Implications of the results to the standardization and deployment of LTE in unlicensed bands (LTE-U) are also discussed.
    \end{abstract}


    \section{Introduction}

    Hundreds of megahertz of unlicensed spectrum under 10 GHz is currently available and more will likely be allocated in the near future. Unlike licensed frequency bands, an unlicensed band is free to use by anyone as long as some basic constraints are satisfied. The constraints are usually on the transmit power spectral density (PSD). In a few bands in several regions in the world, an additional simple protocol (such as listen-before-talk) also needs to be followed.
    A lot of work on unlicensed spectrum focused on WiFi offloading of cellular data \cite{Pyattaev-2013,Dimatteo-2011,Zhang-2012,Zhou-2016}. In comparison with WiFi, Long Term Evolution (LTE) technology has benefits of high efficiency and robust mobility \cite{Qualcomm-2014, Nokia-2014}. Most operators and vendors believe that LTE in unlicensed spectrum (LTE-U) will seamlessly extend cellular networks and require no separate management as WiFi offloading would.

    Challenging coexistence issues arise with multiple LTE-U and WiFi operators.
    It is a consensus within the 3rd Generation Partnership Project (3GPP) that LTE-U should not disrupt concurrent WiFi services \cite{Nokia-2013,Intel-2014}. More importantly, since every LTE-U operator is incentivized to make the maximum use of the free spectrum, without an effective scheme for cooperation, many operators are likely to suffer from significant interference, leading to severely degraded spectral efficiencies, also known as {\em the tragedy of the commons}~\cite{commons-1968}. Two fundamental questions are addressed in this work: 1) Can intrinsically selfish and strategic operators cooperate for their mutual benefit? and 2) if so, how should strategic operators with dynamic traffic cooperate?

    Given the sophisticated nature of wireless operators, it is natural to model them as strategic players and cast the spectrum sharing problem in the framework of game theory and mechanism design.
    There have been some game-theoretic studies of spectrum sharing among non-cooperative parties (e.g.,~\cite{clemens2005intelligent, huang2006distributed, Huang-2006, Tse-2007, Xu-2014,Acharya-2007,Gandhi-2007,Cao-2005,Ji-2006,Kamal-2009, Singh-2015-july, Zhang-Cai-2015}). In particular, \cite{Tse-2007} laid the groundwork in a limited scenario where the number of operators is fixed, the utility function is the Shannon rate, and each operator is subject to a total transmit power constraint (in lieu of PSD constraints).
    The schemes in \cite{Tse-2007} and \cite{Xu-2014} also need each operator to measure the exact power spectral profile of every other operator, which is hard to implement in practice.
    The authors of \cite{Zhang-Cai-2015} devised a pricing mechanism for multiple operators to negotiate power usage in unlicensed spectrum.
    Spectrum sharing among operators having similar rights for accessing spectrum was studied in \cite{Singh-2015-july}, where an internal virtual currency in radio access network was used.

    In this paper, we study coexistence of multiple non-cooperative operators with time-varying traffic. Such traffic variations are likely to be quite pronounced in densely deployed small cells~\cite{Cavalcanti-2005}.
    A class of general utility functions is considered, with Shannon rate being a special case. We start with sharing schemes for a simple scenario where a given fixed set of operators are colocated and orthogonal spectrum sharing is preferred due to higher overall spectral efficiency.
    We first establish the effectiveness of a simple static sharing scheme where each operator's spectrum use does not vary with traffic levels, and show that the proposed profile is a subgame perfect Nash equilibrium (SPNE). Then, assuming operators arrive sequentially, we show that the number of operators willing to invest a fixed cost in order to share the spectrum is limited and depends on the investment cost and externalities. The study of static sharing is a straightforward extension of our prior work \cite{Static-2014}.

    With the total network revenue in mind, we then introduce a dynamic sharing scheme that adapts to the operators' traffic conditions.
    The sharing problem is formulated as a repeated game with private information and communication.
    We devise a dynamic sharing profile, where the operators share information about their traffic intensities, and in any slot, operators with low traffic intensities loan spectrum to those operators with high traffic intensities (with anticipation that borrowers will reciprocate).
    The proposed profile is shown to be a perfect Bayesian equilibrium with truthful reporting of traffic intensities. The practical implication is that all LTE-U operators are likely to enter such a mutually beneficial sharing agreement.

    This study aims to provide a theoretical foundation for LTE-U standardization and deployment. The proposed spectrum sharing schemes are simple and may serve as the basis for a practical design. Devices generally do not have intelligence and are not strategic, thus the choice of the actual equilibrium allocation is likely to be negotiated among the operators via standards bodies such as the 3GPP and WiFi Alliance. All operators must agree to the selected utility vector in advance and follow the profiles according to standards.\footnote{The allocations to different operators may or may not be equal. It is also conceivable to allow devices certain degrees of freedom to negotiate the utility vector.}
    A credible punishment scheme and/or a violation reporting mechanism is needed (and can be implemented in devices) to deter deviation from mutually beneficial coexistence. However, the punishment state will never be visited if all operators consistently comply with the proposed schemes.

    The remainder of the paper is organized as follows. The basic system model is presented in Section~\ref{System Model} and a static sharing scheme is proposed in Section~\ref{Static Sharing Schemes }. We study dynamic sharing schemes that allow operators to trade spectrum in Section~\ref{Dynamic Sharing Schemes}.
    Numerical results are shown in Section~\ref{Ch:unlicensed_static SE:Numerical Results}. Concluding remarks are given in Section~\ref{Conclusion}.

    \section{System Model} \label{System Model}
    Assume $n$ operators share a certain band or bands of unlicensed spectrum taken as a bounded real number set $\mathcal{S}\subset \mathcal{R}$, which is in general a union of some finite intervals. The total amount of spectrum is $W=\int_\mathcal{S} 1 \diff f$ Hz. On the timescale of interest, it is conveniently assumed that the spectrum in $\mathcal{S}$ is homogeneous. It is also assumed that the operators can monitor each other's PSD.\footnote{It could be implemented by using LTE CRS (cell reference signal), which embeds the cell identity and could separate the average received power of cells. Alternatively, an operator may infer about other operators' usage based on its own quality of service without active monitoring. This is not considered here.}

    We focus on a discrete-time formulation, where time is slotted and all operators are fully synchronized at the slot level.
    The traffic intensity of operator $i$ in slot $t$ is defined as the traffic level during slot $t$, and the intensity is revealed only to operator $i$ at the beginning of slot $t$.
    Each operator's traffic condition is private information.
    We assume that the traffic intensity of operator $i$ is an exogenous random process denoted by $\{\Lambda_t^i, t\geq 0\}$ independent of the other operators.

    Each operator determines its transmit PSD at the beginning of each slot based on its prior information, and maintains the same PSD over the entire slot.\footnote{In a more nuanced setting, the selected PSD is a mask that constrains the actual PSD which can vary over a slot.}
    For ease of notation, denote the transmit PSD of operator $i$ in slot $t$, normalized by the flat noise PSD, as $p^i_t(f)$.
    Throughout this paper, it is assumed that the only exogenous constraint on the PSD is regulatory, where the transmit PSD must be upper bounded by $P$, i.e., $p_t^i(f) \leq P$ for all $i$, $t$, and $f$.
    In general, the utility of operator $i$ is a function of the PSDs chosen by all operators, as well as its own traffic intensity $\lambda_t^i$, denoted as $u^i(p_t^i, p_t^{-i}, \lambda_t^i)$, where $p_t^{-i}$ denotes collectively the PSDs of operators other than $i$.
    The total revenue of operator $i$ over the infinite time horizon is defined as
    \begin{align} \label{eq:weighted_utility}
     V^i=(1-\delta)\sum_{t=0}^{\infty} \delta^{t} u^i(p_t^i, p_t^{-i}, \lambda_t^i)
    \end{align}
    where $\delta \in [0,1)$ denotes the discount of future utility, and the factor $(1-\delta)$ makes it convenient to compare the total revenue with a one-slot utility. In practice, an operator is concerned with the utilities over the course of many slots (at least days), hence the discount factor is typically very close to 1.
    %{\color{red}The cost (e.g., energy) of using PSD $p^i_t$ can be easily incorporated into the utility function.}

    It is common in practice that operators deploy their transmitters on the same tower or the towers close to each other. For simplicity, we assume for most of our discussion that the cells of all operators completely overlap, all transmitters are colocated, and all receivers are one unit distance away from the transmitters. The key principles here apply to more general scenarios with partially overlapped cells \cite{Thesis-2016}.
    For operator $i$, the signal-to-interference-and-noise ratio (SINR) in slot $t$ at frequency $f$ is expressed as
    \begin{align}
    \gamma_t^i(f)= \frac{ p^i_t(f)}{1+ \sum_{j\neq i} p^j_t(f)}.
    \end{align}
    In general, let the ``usefulness'' per Hz at the vicinity of a certain frequency $f$ be $r(\gamma^i(f))$, where $r(\cdot)$ is strictly increasing (hence increasing the SINR makes the frequency more useful).
    In this paper, we also assume that the system operates in the interference-limited regime. In particular, $r(\cdot)$ satisfies:
    \begin{align} \label{utility_property}
    r\left(P\right)> \sup_{n\geq 2} \left[ n\cdot r \left(\frac{P}{(n-1)P+1}\right) \right].
    \end{align}
    Condition \eqref{utility_property} implies that the sum utility that multiple operators extract by simultaneously sharing the same piece of spectrum with maximum transmit power is less than that can be extracted from exclusive use of the spectrum.
    A practically relevant example of $r(\cdot)$ is $r(\gamma)=\log(1+\gamma)$ so that $\int_\mathcal{S} r(\gamma_t^i(f))\diff f$ is the Shannon capacity of the additive white Gaussian noise channel. In that case, it is easy to show that if $P>1.62$, then \eqref{utility_property} is always satisfied.

    For concreteness, we introduce a family of utility functions that depend on the accumulated spectrum usefulness and the traffic intensity.
    Let the utility of operator $i$ in slot $t$ take the following form:
    \begin{align} \label{def_utility}
    u^i(p_t^i, p_t^{-i}, \lambda_t^i) = \pi^i\left( \frac1{r( P)} \int_\mathcal{S} r(\gamma_t^i(f))\diff f, \lambda_t^i\right)
    \end{align}
    which is bounded and non-negative.
    If the spectral efficiency on occupied spectrum is homogeneous, then the integral in \eqref{def_utility} can be simplified by using the bandwidths.
    Under orthogonal sharing, operator $i$ maximizes its utility by transmitting at peak power over its spectrum, yielding utility $\pi^i(w_t^i,\lambda_t^i)$, where $w_t^i$ is the bandwidth occupied by operator $i$ in slot $t$.
    Indeed, the first argument on the right side of \eqref{def_utility} can be viewed as the effective exclusive bandwidth occupied by operator $i$ that yields the same utility.

    In this paper, we consider a class of functions $\pi^i$ satisfying some additional conditions:
    \begin{itemize}
    \item[a)] $\pi^i(x, \lambda)$ is continuous, strictly increasing, and strictly concave in $x$ for every $\lambda$.
      As a consequence, the incremental utility of adding spectrum usefulness decreases as the initial amount of spectrum usefulness increases.
      Precisely, for every $x,\lambda$, and $\Delta>0$,
      \begin{align}\label{property_a}
        \pi^i(x+\Delta, \lambda)- \pi^i(x,\lambda)< \pi^i(x, \lambda)- \pi^i(x-\Delta, \lambda).
      \end{align}
    \item[b)]$\pi^i(x, \lambda)$ is finite and strictly supermodular, that is, adding an incremental amount of spectrum usefulness yields higher improvement in the utility when the traffic intensity is higher.  Precisely, for every $x$, $\xi > \lambda$, and $\Delta>0$,
      \begin{align} \label{property_b}
        \pi^i(x+\Delta, \lambda)-\pi^i(x, \lambda) < \pi^i(x+\Delta, \xi)-\pi^i(x, \xi).
      \end{align}
    \end{itemize}

    Without coordination, each operator would prefer to transmit over the entire spectrum using the maximum power, referred to as {\em full-spectrum strategy}.
    Clearly, if all operators employ the full-spectrum strategy (the strategy profile is referred to as {\em full-spectrum sharing}), they all suffer the maximum interference, which results in poor spectral efficiency. In particular, under full-spectrum sharing, the utility of operator $i$ in slot $t$ is
    \begin{align}
    \pi^i_f(\lambda_t^i)= \pi^i\left(\frac{W}{r(P)}r\left(\frac{P}{P(n-1)+1}\right),\lambda_t^i\right),
    \end{align}
     and the expected utility of operator $i$ is
    \begin{align} \label{eq:u_f}
     u_f^i = \Exp \left[\pi^i_f(\Lambda_t^i)\right].
    \end{align}
    Alternatively, if the operators avoid interfering with each other by using different parts of the spectrum, the spectral efficiency and the sum utility become much higher. This suggests that it may be beneficial for strategic operators to cooperate.

    This paper addresses the general situation where operators' traffic intensities vary over time.  Shorter slot duration leads to better tracking of traffic intensities, which results in higher revenues.
    However, due to the limitations of techniques in spectrum monitoring and information exchange among operators, the slot duration cannot be arbitrarily small. The slot duration should be chosen to balance the benefit of spectrum agility with the cost, robustness, and other practical issues.  Throughout this paper, a fixed slot duration is assumed, which is conceived to be between a few seconds to several minutes.

\section{Static Sharing Schemes } \label{Static Sharing Schemes }

    In this section, we study static sharing schemes, where the PSD of an operator does not vary with the traffic dynamics (even though the utilities do). The history of the realized actions (the chosen PSDs) provides full information for choosing the strategy in the current slot. In particular, the history does not affect the set of admissible actions or utilities in the future, and the strategy of an operator is actually independent of the traffic intensities. In this context, the spectrum sharing problem can be modeled as a repeated game with complete and perfect information over an infinite sequence of slots.
    We begin with the simplest case of two operators in Section~\ref{Static Sharing Schemes_2}, then generalize to the case of an arbitrary fixed number of operators in Section~\ref{Static Sharing Schemes_n}. In Section~\ref{Static Sharing Schemes_entry}, we study the entry game, where an arbitrary number of users arrive sequentially.

    The following simple observation is useful.
   \begin{lemma} \label{Orth_vs_Full}
   If \eqref{utility_property} holds, equal, orthogonal sharing achieves higher utility than full-spectrum sharing for all operators.
   \end{lemma}
   \begin{proof}
    Under equal, orthogonal sharing, the expected utility of operator $i$ in one slot is
    ${u_o^i = \Exp \left[\pi^i\left(\frac{W}{n}, \Lambda_t^i\right)\right]}$.
    By assumption~\eqref{utility_property},
    \begin{align}  \label{eq:2}
        W r\left( \frac{P}{1+ (n-1)P}\right) < \frac{W}{n} r\left(P\right).
    \end{align}
    Since $\pi^i(x, \lambda)$ is increasing in $x$ for every $\lambda$, we have $\pi^i\left(\frac{W}{n}, \lambda\right)> \pi^i\left(\frac{W}{r(P)}r\left(\frac{P}{P(n-1)+1}\right), \lambda\right) $ for every $\lambda$. Hence, $u_o^i$ is larger than $u_f^i$ given in \eqref{eq:u_f}.
    \end{proof}

    \subsection{The Two-Operator Case} \label{Static Sharing Schemes_2}

    First, consider a system with two operators who only operate for one slot with fixed PSDs.
    This is referred to as the stage game.
    The action space $\mathcal{P}$ is the set of feasible PSDs.
    A strategy profile $(p^1, p^2)\in \mathcal{P}\times \mathcal{P} $ is a strict Nash equilibrium of the stage game if an operator becomes worse off by unilaterally deviating from it, i.e.,
    $ u^1(p^1, p^{2},\lambda^1) > u^1(q, p^{2}, \lambda^1)$ and $u^2(p^2, p^{1},\lambda^2) > u^2(q, p^1,\lambda^2)$ for every $\lambda^1$, $\lambda^2$, and $ q\in  \mathcal{P}$.
    The minimax utility of an operator is defined as the smallest utility that the other operator can force it to receive, regardless of its strategy.

    \begin{lemma} \label{flat-psd}
    In the stage game, all operators using the full-spectrum strategy is the unique (strict) Nash equilibrium. In particular, $\pi^i_f(\lambda^i)$ is the minimax utility of operator $i$ with traffic~$\lambda^i$.
    \end{lemma}
    The proof is trivial because an operator's utility increases as its PSD increases.

    In practice, the stage game is repeated an infinite number of times and each operator can vary its PSD over time.
    The players are the two operators; the action space $\mathcal{P}$ is the set of feasible PSDs; at the end of each epoch, the operators can observe the action of the other operator and can use the complete history of play to decide on future actions; the strategy space of an operator is the set of complete plans of actions that define what PSD the operator will use in every possible event where the operator needs to act; and the payoff is the expected total revenue over the infinite time horizon expressed in \eqref{eq:weighted_utility}.
    As was observed in~\cite{Tse-2007}, the repeated game allows a much richer set of Nash equilibria than the stage game.

    In this section, we consider profiles which constitute an SPNE. SPNE is a refined and stronger notion than Nash equilibrium \cite{Fundenberg-1986}, in the sense that such equilibria cannot be merely the consequence of non-credible threats, and thus are rationale and likely outcomes in practice.
    In particular, a strategy profile is an SPNE if the restriction of the strategy profile yields a Nash equilibrium from the start of each stage for each history~\cite{Fundenberg-1991}.
    In order to verify a profile as an SPNE, we can use the one-shot deviation principle, which addresses the equivalence between single-deviation optimality and strategy-deviation optimality.
    In particular, in an infinite-horizon multi-stage game with perfect information, a strategy profile is an SPNE if and only if no operator can gain by deviating from the profile in a single stage and conforming to the profile thereafter \cite[Theorem 4.2]{Fundenberg-1991}.

    \begin{lemma}  \label{folk-thm}
    Let $(u^1,u^2)$ be a feasible utility pair, i.e., there exists a pair of PSDs the two operators can adopt to attain $(u^1,u^2)$ as their utilities. If $(u^1, u^2)> (u_f^1, u_f^2)$, then there is an SPNE where the corresponding expected total revenue pair is $(u^1, u^2)$ as long as the future discount factor $\delta$ is sufficiently close to 1.
    \end{lemma}

    \begin{proof}
     We construct such an SPNE. Let the utility pair $(u^1,u^2)$ be attained by the PSD pair $(p^1,p^2)$, i.e., $\Exp \left[u^1(p^1,p^2,\Lambda_t^1)\right]=u^1$ and $\Exp \left[u^2(p^2,p^1,\Lambda_t^2)\right]=u^2$.

    \begin{profile} \label{mech_full_forever}
     The strategy of each operator $i\in\{1,2\}$ is: Use PSD $p^i$ in slot 0 and continue to use $p^i$ in each subsequent slot, as long as the other operator $i'$ ($i'\neq i$) uses PSD $p^{i'}$ in the previous slot; otherwise, transmit the maximum power over the entire spectrum hereafter.
    \end{profile}

    It is convenient to refer to the slots in which the operators use $(p^1,p^2)$ as the {\em cooperation state}, and the remaining slots (if any) as the {\em punishment state}.
     The operators begin with the cooperation state. Once an operator deviates from the cooperation state, the punishment is everlasting, so that the operator suffers a net loss in revenue. Moveover, transmitting the maximum power over the entire spectrum is the best response in the punishment state. Hence one-shot deviation is undesirable and Lemma \ref{folk-thm} holds.
    \end{proof}

    The cooperation defined in Profile~\ref{mech_full_forever} is not robust in practice. In case of any perceived deviation, even if due to false detection of the other operator's spectrum usage, the operators will be trapped in the punishment state.
    We next consider a profile that provides incentives for deviating operators to return to the cooperation state. The idea is to return to the cooperation state after spending sufficiently many slots in the punishment state.
    Denote
    \begin{align} \label{eq_upper_utility}
    \overline{U}^i(\lambda)=\sup_{p,\widehat{p}} u^i(p,\widehat{p}, \lambda).
    \end{align}
    Choose the duration of punishment $T$ such that for every $\lambda$ and $i=1,2$,
    \begin{align}\label{eq_punish}
    \overline{U}^i(\lambda)- \pi^i\left(\frac{W}{2},\lambda\right)< T(u^i-u_f^i),
    \end{align}
    which implies that one's loss due to punishment is greater than the one-shot gain by deviating if there is no future discount.
    To reduce the length of punishment, in practice we choose the minimum $T$ that satisfies \eqref{eq_punish}.
    \begin{profile} \label{mech_full_slot}
    Operator $i$ uses $p^i$ in slot 0.  In each subsequent slot, the system state evolution and the strategy of operator $i\in\{1,2\}$ depends on the state as follows:
    \begin{itemize}
    \item[I.] Cooperation state: If the PSDs were $(p^1, p^2)$ in the previous slot, then use $p^i$ in this slot; otherwise, transit to the punishment state.
    \item[II.] Punishment state: If any operator is detected to deviate from this profile in any state in the previous slot, then use the maximum power over the entire spectrum for $T$ slots and then resume using $p^i$ and return to the cooperation state.
    \end{itemize}
    \end{profile}

    \begin{lemma}  \label{folk-thm-desired}
    For any feasible vector $(u^1, u^2)> (u_f^1, u_f^2)$, Profile~\ref{mech_full_slot} is an SPNE that achieves the expected total revenue pair $(u^1, u^2)$, as long as the punishment duration $T$ satisfies \eqref{eq_punish} and the future discount factor $\delta$ is sufficiently close to 1.
    \end{lemma}
   The proof is similar to Theorem 2 in \cite{Fundenberg-1986} for general repeated games.
   According to Lemmas~\ref{folk-thm} and~\ref{folk-thm-desired}, in order to achieve the total revenue vector $(u^1, u^2)> (u_f^1, u_f^2)$, the operators announce their strategies and agree to the selected utility vector in advance.

   In general, there are an infinite number of equilibria, and hence an infinite number of equilibrium utility tuples. In practice, the operators need to negotiate a favorable equilibrium allocation, usually through some standard bodies.
   %The choice of the actual equilibrium allocation must be agreed by all operators.
   Once all operators agree on the equilibrium, they monitor each other's PSDs and respond accordingly without calculating \eqref{def_utility}.
   Although an equilibrium allocation need not be orthogonal, an implication of Lemma~\ref{Orth_vs_Full} is that an orthogonal allocation is in general favorable. In addition, it is much easier to verify orthogonal sharing by monitoring the PSD support than the exact PSD of another operator. Furthermore, orthogonal sharing allows an operator to use its share of the spectrum as if it were licensed. Therefore, hereafter we consider orthogonal sharing. Without loss of generality, we also assume the chosen equilibrium corresponds to equal partition of spectrum by the operators, even if their one-shot utilities and traffic statistics are different. The key principles in this paper apply to general equilibrium allocations.

   Denote a set of PSDs corresponding to equal, orthogonal sharing as $(p_o^j)_{j=1}^n$ and let ${u_o^i = \Exp \left[u^i(p_o^i, p_o^{-i},\Lambda_t^i)\right]}$.
   When $n=2$, $(u_o^1, u_o^2)$ dominates $(u_f^1, u_f^2)$. From Lemma~\ref{folk-thm-desired}, the expected total revenue pair $(u_o^1, u_o^2)$ can be achieved by a Nash equilibrium in the repeated game, as depicted in Fig. \ref{fig:illustration_static}. According to the corresponding strategy, each operator transmits in an exclusive half of the unlicensed spectrum and has no incentive to deviate.
   %\footnote{It is easy to introduce a protocol to determine which values the operators choose to use.}

    \begin{figure}
    \centering
    \includegraphics[width=.8\columnwidth]{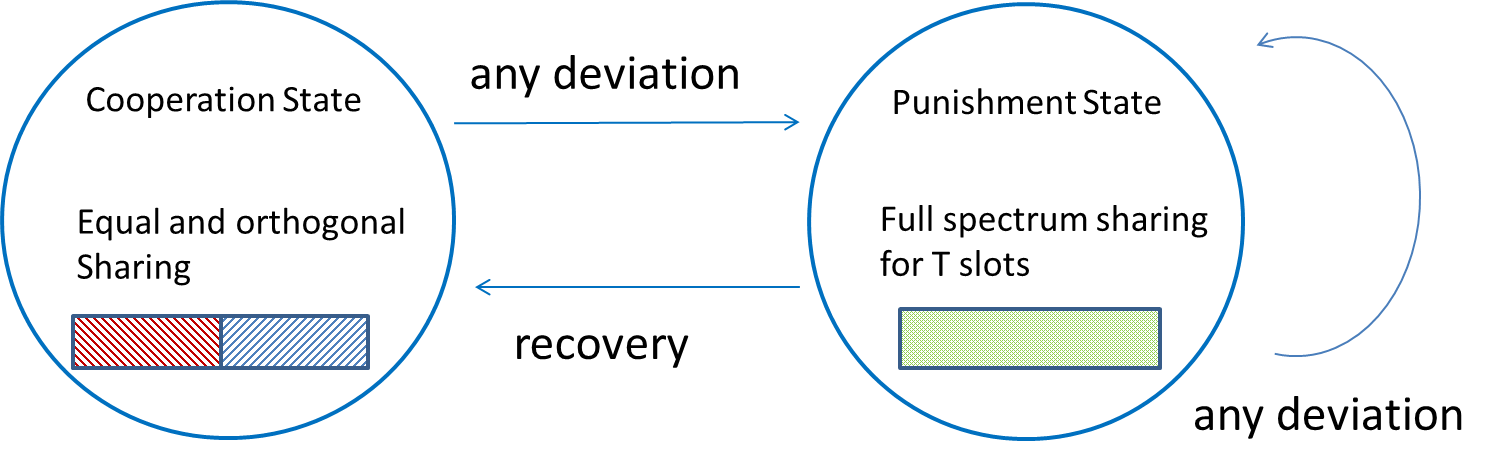}
    \caption{Static spectrum sharing with equal and orthogonal partitions of spectrum in the cooperation state.}
    \label{fig:illustration_static}
    \end{figure}

   \subsection{The $n$-Operator Case} \label{Static Sharing Schemes_n}
   The preceding sharing schemes can be easily extended to the case of $n$ operators.
   According to Lemma~\ref{Orth_vs_Full}, we have $u_o^i> u_f^i$.
   Define the maximum utility of operator $i$ with traffic intensity~$\lambda$ in one slot as $\overline{U}^i(\lambda)=\sup_{p,p^{-i}} u^i(p,p^{-i}, \lambda)$.
   We denote the equal-sharing bandwidth as
   \begin{align}
   w=\frac{W}{n}.
   \end{align}
   Then choose $T$ such that for every $i$ and $\lambda$,
    \begin{align}\label{eq_punish_n}
    \overline{U}^i(\lambda)-\pi^i\left(w, \lambda\right)
    < T\left(u_o^i-u_f^i\right).
    \end{align}
   Consider the following straightforward generalization of Profile~\ref{mech_full_slot}.
    \begin{profile} \label{mech_full_slot_n}
    Operator $i$ uses $p^i$ in slot 0.  In each subsequent slot, the system state evolution and the strategy of operator $i\in\{1,\ldots,n\}$ depends on the state as follows:
    \begin{itemize}
    \item[I.] Cooperation state: If the PSDs were $(p^1,\ldots, p^n)$ in the previous slot, then use $p^i$ in this slot; otherwise, transit to the punishment state.
    \item[II.] Punishment state: If any operator is detected to deviate from this profile in any state in the previous slot, then use the maximum power over the entire spectrum for $T$ slots and then resume using $p^i$ and return to the cooperation state.
    \end{itemize}
    \end{profile}
   \begin{lemma} \label{folk_thm_n}
   Profile~\ref{mech_full_slot_n} is an SPNE for $n$ operators as long as \eqref{utility_property} and \eqref{eq_punish_n} hold and the future discount factor $\delta$ is sufficiently close to 1.
   \end{lemma}
   The proof is a simple generalization of the 2-operator case (Lemma \ref{folk-thm-desired}) and is omitted.

   Lemma~\ref{folk_thm_n} implies that the operators have no incentive to deviate from equal, orthogonal spectrum sharing. It is not difficult to show that if the operators agree on an unequal allocation, equilibria can also be achieved under Profile~\ref{mech_full_slot_n}.

   \subsection{Entry Problem} \label{Static Sharing Schemes_entry}

   The preceding discussions are based on the assumption that a given fixed number of operators share the spectrum in a given area.
   We now study the situation where an arbitrary number of strategic operators may arrive in a sequential manner. An operator must make an investment (e.g., on network infrastructure) before using any spectrum.
   For simplicity, we assume all operators have identical PSD constraints, utility functions, investment costs (assumed to be $c$) and traffic intensity distributions, so the superscript will be dropped in this subsection. The analysis can be easily extended to the case of heterogeneous operators.
   If $n$ operators have invested to share the spectrum, the expected utility of each operator in one slot is
   \begin{align}
    u_f(n)=\Exp \left[\pi\left(\frac{W}{r(P)}r\left(\frac{P}{P(n-1)+1}\right), \Lambda_t\right)\right]
   \end{align}
   under full-spectrum sharing or is
       \begin{align}
            u_o(n)=\Exp \left[\pi\left(w r(P), \Lambda_t\right)\right]
       \end{align}
   under equal and orthogonal sharing.
   The difference here is that an incumbent operator may change its action upon investment by a new operator.  It is natural for each incumbent operator to choose between two actions, `to punish' and `to cooperate', when a new operator begins to use any spectrum. If all operators use full spectrum to punish the new operator, then everyone achieves the utility of $u_f(n+1)$; if all cooperate with the new operator by using $\frac{1}{n+1}$ of the spectrum, then each achieves the utility of $u_o(n+1)$.

    According to Lemma~\ref{Orth_vs_Full}, we have $u_o(n)-u_f(n)\geq 0$ for every $n$. It is easy to verify that both $u_f(n)$ and $u_o(n)$ vanish as $n\rightarrow \infty$. Therefore, if $u_f(1)\geq c$, there exists $n^*$ such that $u_f(n^*+1)<c$ and $u_f(n^*)\geq c$.
    If $u_f(1)< c$, we let $n^*=0$. Evidently, larger $c$ results in smaller~$n^*$.
    Denote $\overline{U}(\lambda)$ as the maximum utility of an operator with traffic $\lambda$, which is achieved if the operator uses the full-spectrum strategy while other operators make no transmissions. Let $T(n)$ be such that for every $\lambda$,
    \begin{align}\label{eq_punish_entry}
    \overline{U}(\lambda) - \pi\left(w,\lambda\right) < T(n)\left(u_o(n)-u_f(n)\right).
    \end{align}
    \begin{profile} \label{mech_entry_game}
    The $i$-th operator to arrive does not invest if $i>n^*$.  If $i\le n^*$, the operator invests and performs the following in each slot thereafter:
    \begin{itemize}
      \item If there are $n\leq n^*$ active operators, use Profile~\ref{mech_full_slot_n} for $n$ operators.
      \item If there are more than $n^*$ operators, always use the full-spectrum strategy.
    \end{itemize}
    \end{profile}

    \begin{theorem}  \label{entry_game}
    Profile~\ref{mech_entry_game} is an SPNE, as long as \eqref{utility_property} and \eqref{eq_punish_entry} hold and the future discount factor~$\delta$ is sufficiently close to 1.
    \end{theorem}
    \begin{proof}
    By using the one-shot deviation principle in \cite[Theorem 4.2]{Fundenberg-1991}, it suffices to show that if an operator deviates from Profile~\ref{mech_entry_game} in a single slot and then returns to conform to Profile~\ref{mech_entry_game}, the operator suffers a net loss.
    For incumbent operators, first consider the one-shot deviation in the cooperation state for $n\leq n^*$. Assuming an operator deviates at slot $t$ in the cooperation state and conforms afterwards, the punishment starts at slot $t+1$. If there is no new entrant during the punishment slots, then the utilities are identical to the case with a fixed number of operators. Hence the deviating operator's utility decreases according to Lemma~\ref{folk_thm_n}.
    On the other hand, if a new operator starts transmission before the punishment ends, the new operator transmits the maximum power over the entire spectrum, further reducing the payoff of the deviator.
    The same arguments apply to the case where multiple new operators arrive before returning to cooperation.
    Moreover, the preceding analysis applies to the situation where the new operator $i\leq n^*$ starts transmission but deviates by using a different PSD than required by the current state.

    Since the minimax utility for incumbent operator with traffic $\lambda$ is $\pi_f(\lambda)$, deviation in the punishment state only lengthens the punishment and postpones the larger utility. For new entrant~$i$, if $i\leq n^*$ (resp.,\ $i>n^*$), the total revenue is greater (resp., less) than the investment cost. Hence, (one stage) deviation from Profile~\ref{mech_entry_game}'s investment decision is not profitable. Therefore, Profile~\ref{mech_entry_game} is an SPNE.
    \end{proof}

    According to Theorem~\ref{entry_game}, under the proposed profile, there will be at most $n^*$ active operators in the market and all of operators obtain larger revenue than full-spectrum sharing.
    The proposed profile provides a way to share the spectrum efficiently when there are an indefinite number of strategic operators, thereby mitigating the effects of the tragedy of the commons.

    \section{Dynamic Sharing Schemes} \label{Dynamic Sharing Schemes}

    The static sharing strategies discussed in Section~\ref{Static Sharing Schemes } are in general not the most efficient under dynamic traffic conditions.  At any given time, some operators may have light traffic and hence excess spectrum, while others may have heavy traffic and hence experience a spectrum shortage.
    In this section, we allow the operators to trade spectrum and adapt their spectrum usage to the traffic conditions in each slot.
    To be consistent with the previous section, we assume that when no trade happens, operators agree to equal, orthogonal sharing. The main principles here apply to general unequal sharing.

    Spectrum trade can take place either with or without monetary payment. Spectrum sharing without monetary payment has a few advantages in some aspects. In comparison with monetary trading, sharing without monetary payment may suppress spectrum ``trolls''--operators who serve few or no customer, yet demand payment so as to not cause interference. Besides, sharing without monetary payment avoids pricing, metering, and billing for spectrum usage, which require efforts beyond the physical layer and can be costly. This study is hence restricted to spectrum trade without monetary payment. The basic idea is that an operator in need of extra spectrum may borrow from another operator but needs to return the spectrum in the future.

    A sharing scheme can be either direct or indirect. With a direct scheme, all operators report their traffic intensities, and each operator uses a designated strategy according to all reported traffic intensities. With an indirect scheme, each operator may report other signals (i.e., desired actions) instead of traffic intensities, and then determine the spectrum to utilize according to one's own traffic intensity and other operators¡¯ reports.
    By the Revelation Principle of Bayesian games \cite[Theorem 2]{Myerson-1979}, for every Nash equilibrium of an indirect scheme, there exists a direct scheme that is payoff-equivalent and in which truthful reporting is a Nash equilibrium. Thus, in this paper we focus on direct schemes where operators report their traffic intensities.

    The spectrum sharing problem with dynamic traffic is formulated as a repeated game with private information and communication.
    First, consider the stage game over a single slot. At the beginning of the game, the traffic intensity of operator $i$ is randomly generated. Each operator only knows its own traffic level. The operators report their traffic intensities (either truthfully or not) through a common collaborative channel at the beginning of the slot.
    After the operators receive all reported traffic intensities, each operator chooses a PSD for transmission. The one-shot utility for operator $i$ is defined the same as \eqref{def_utility}.
    An operator's strategy consists of its reported traffic intensity and its PSD for transmission over the slot.
    In analogy with Lemma~\ref{flat-psd}, it is easy to show that the Nash equilibrium of the stage game is full-spectrum sharing.

    In practice, the stage game with incomplete information is repeated infinitely.
    In each slot, the operators adjust their PSDs based on the reported traffic intensities and the history of realized actions.
    The strategy space of an operator is the set of complete plans of actions that define what traffic intensity the operator will report and what PSD the operator will use in every possible event, and the payoff is the expected total revenue in \eqref{eq:weighted_utility}.

    The preceding game is an infinite game with incomplete information, due to possible untruthful reporting.
    We shall show that the proposed strategy profiles are {\em perfect Bayesian equilibria}.
    Perfect Bayesian equilibrium is a stronger notion than SPNE because in addition to being subgame perfect, the equilibrium must be such that all operators expect other operators to continue to play according to their respective equilibrium strategies even after some operators deviate from the equilibrium path \cite{Fundenberg-1991}. In other words, the strategies must be credible in every step of the continuation game.
    Fudenberg and Tirole \cite{Fundenberg-PBE-1991} provided the leading formal definition of perfect Bayesian equilibrium, which applies to {\em finite} multi-stage games with observed actions and independent private information.
    Several variants of perfect Bayesian equilibrium for finite games ensued \cite{Battigalli-1996,Bonanno-2013,Conzalez-2014}.
    For infinite games, perfect Bayesian equilibrium was defined in~\cite{Watson-2016} by generalizing the reasonableness condition in finite games with observed actions in~\cite{Fundenberg-PBE-1991}.

    In search of an efficient perfect Bayesian equilibrium, we seek a direct scheme such that it is in the best interest of every operator to truthfully report their traffic for spectrum borrowing/lending.
    Denote the net spectrum balance of operator $i$ at the beginning of slot $t$ as $b_t^i$. The assumption $b_t^i > 0$ (resp., $b_t^i<0$) means that by the end of slot $t-1$ operator $i$ has lent more (resp., less) spectrum than borrowed. For simplicity, we set a balance constraint as $|b_t^i|\leq \overline{b}$ for all $i$ and $t$ to preclude infinite borrowing (e.g., in the manner of a Ponzi scheme).

    Throughout this section, we assume that in each slot, the traffic intensity of operator $i$ is either high $(\Lambda_t^i = 1)$ or low $(\Lambda_t^i = 0)$.
    It is further assumed that $(\Lambda_t^i)_{i,t}$ are independent and identically distributed for every $i$.
    The two-level traffic assumption is meaningful in practice where the sharing protocol needs to be made simple. The main principles developed here can be extended to more general traffic assumptions.

    \subsection{The Two-Operator Case} \label{Dynamic Sharing Schemes_2op}

    We begin with the two-operator case.
    Let $\Delta\in(0,w]$ be the amount of spectrum that the operator with high traffic intensity borrows from the operator with low traffic intensity in one slot.
    Consider the following strategy profile.

    \begin{profile} \label{mech_1}
    The strategies of the two operators mirror each other.
    Let the system start from the cooperation state with beginning balances $b^1_0=b^2_0=0$.
    Operator 1's strategy in slot $t$ depends on the state as follows:
    \begin{itemize}
    \item[I.] Cooperation state:
    \begin{itemize}
    \item[a.] Reveal its own traffic intensity $\lambda_t^1$, and learn $\lambda_t^2$ from operator 2.
    \item[b.] If $\lambda_t^1>\lambda_t^2$ and $b_t^1-\Delta>=-\bar{b}$, then use $w+\Delta$ Hz for transmission and reduce the balance to $b_{t+1}^1=b_t^1-\Delta$; if $\lambda_t^1<\lambda_t^2$ and $b_t^1+\Delta<=\bar{b}$, then use $w-\Delta$ Hz for transmission and increase the balance to $b_{t+1}^1=b_t^1+\Delta$;
        otherwise, let $b_{t+1}^1=b_t^1$ and use $w$ Hz for transmission.
    \end{itemize}
    \item[II.] Punishment state: If any operator is detected to deviate from this profile in any state in the previous slot, then use the maximum power over the entire spectrum for $T$ slots and then return to the cooperation state. The balances remain unchanged during the punishment state.
    \end{itemize}
    \end{profile}

    Because spectrum trade occurs only if there is sufficient balance, the balance $b_t^i$ remains within $[-\overline{b},\overline{b}]$ at all times.
    Without loss of generality, let $\overline{b}=k \Delta$ where $k$ is a positive integer.
    There are two types of deviations: undetectable deviation, i.e., lying about one's own traffic, and detectable deviation, i.e., using a different amount of spectrum than dictated by the profile.

    We shall show that Profile~\ref{mech_1} constitutes a perfect Bayesian equilibrium as defined in \cite{Watson-2016}.
    An operator's {\em information set} is a set that establishes all the possible moves that could have taken place in the game so far, given what the operator has observed \cite{Fundenberg-1991}. If the game has perfect information, every information set contains only one element, namely the point actually reached at that stage of the game. Otherwise, it is the case that some operators cannot be exactly sure about some unobserved variables in the game so far.  Denote $\mathcal{H}$ as the collection of all information sets.
    A perfect Bayesian equilibrium includes a strategy profile which describe actions at every information set, and a system of beliefs on strategy profiles which captures the operator's conjecture about both how each information set is reached and what will happen from the information set.
    The system of beliefs $\pmb{\mu}$ is defined as the collection of beliefs $\mu(h)$, for every $h\in \mathcal{H}$, where the belief $\mu(h)$ is a distribution on pure strategy profiles with support on the set of pure strategy profiles that reach information set $h$.

    Provided that Profile~\ref{mech_1} is used, the proper system of beliefs $\pmb{\mu}$ is constructed as follows: The beliefs at an information set are concentrated on the strategy profiles with the actions that have been observed and the prescribed actions at information sets that have not yet been observed.
    In particular, the belief on the profiles with the prescribed actions at unobserved information sets must be consistent with traffic intensity distribution.

    \begin{theorem} \label{thm_dynamic}
     There exist $\Delta>0$ and $T$ such that Profile~\ref{mech_1} along with the system of beliefs $\pmb{\mu}$ is a perfect Bayesian equilibrium, as long as \eqref{utility_property}, \eqref{property_a} and \eqref{property_b} hold, $\Prob(\Lambda_t^1=0,\Lambda_t^2=1)>0$, $\Prob(\Lambda_t^1=1,\Lambda_t^2=0)>0$, and the future discount factor $\delta$ is sufficiently close to 1.
    \end{theorem}
    To prove Theorem 2, we shall verify the three sufficient conditions for achieving perfect Bayesian equilibria \cite{Watson-2016}, namely, plain consistency ($\pmb{\mu}$ must follow a consistent conditional-probability updating), sequential rationality (at each information set, the operator cannot gain by deviating in one slot and thereafter conforming), and that $\pmb{\mu}$ must conform to Profile~\ref{mech_1}. We relegate the detailed proof of Theorem~\ref{thm_dynamic} to Appendix \ref{appandix:Thm2}.

    \subsection{The $n$-Operator Case}

    The sharing scheme introduced in Section~\ref{Dynamic Sharing Schemes_2op} can be easily extended to the case of multiple operators. A trading policy needs to be set up, so that the borrowers know who the respective lenders are. Assume operators have perfect recall of trading history. This implies that each operator knows the current balances of all operators.

    We will discuss a scheme for $n$ operators sharing $W$ Hz unlicensed spectrum.
    Denote the set of operators who report a high traffic intensity and have balance no less than $-\overline{b}+\Delta$ as $\mathcal{A}_1$, and the set of operators who report a low traffic intensity and have balance no greater than $\overline{b}-\Delta$ as $\mathcal{A}_0$. The trading policy $\mathcal{Q}$ is that the operator with the $i$-th largest balance in $\mathcal{A}_1$ borrows $\Delta$ Hz from the operator with the $i$-th smallest balance in $\mathcal{A}_0$, for any $i \leq \min\{|\mathcal{A}_1|, |\mathcal{A}_0|\}$.
    This pairing and subsequent trade is one possible scheme. In general, we can map the current balances and reported traffic conditions to a desired spectrum allocation, where the key findings remain the same.

    \begin{profile} \label{mech_2}
    Let the system start from the cooperation state with $b_0^i=0$, for all $i$. Operator $i$'s strategy in slot $t$ depends on the state as follows:
    \begin{itemize}
    \item[I.] Cooperation state:
    \begin{itemize}
    \item[a.] Reveal the traffic intensity $\lambda_t^i$ and learn traffic intensities from other operators;
    \item[b.] If chosen to trade by the trading policy $\mathcal{Q}$ and $\lambda_t^i=1$, then use $w+\Delta$ Hz for transmission and let $b_{t+1}^i=b_t^i-\Delta$; if chosen to trade and $\lambda_t^i=0$, then use $w-\Delta$ Hz for transmission and let $b_{t+1}^i=b_t^i+\Delta$;
        if not chosen, let $b_{t+1}^i=b_t^i$ and use $w$ Hz for transmission.
    \end{itemize}
    \item[II.] Punishment state: If any operator is detected to deviate from this profile in any state in the previous slot, then use the maximum power over the entire spectrum for $T$ slots and then return to the cooperation state. The balances remain unchanged during the punishment state.
    \end{itemize}
    \end{profile}

   Provided that Profile~\ref{mech_2} is used,  we construct the proper system of beliefs $\pmb{\mu}$ in the similar way as the two-operator case.

    \begin{theorem} \label{lemma_NOp_twoLev}
    There exist $\Delta>0$ and $T$ such that Profile~\ref{mech_2} along with the system of beliefs $\pmb{\mu}$ is a perfect Bayesian equilibrium, as long as \eqref{utility_property}, \eqref{property_a} and \eqref{property_b} hold, the future discount factor $\delta$ is sufficiently close to 1 and for every $i$, there exists $j$ such that $\Prob(\Lambda_t^i>\Lambda_t^j)>0$.
    \end{theorem}
    Theorem \ref{lemma_NOp_twoLev} is proved in Appendix \ref{appandix:Thm3}.

    In this section, we have shown the existence of perfect Bayesian equilibria with dynamic spectrum sharing.
    Profiles~\ref{mech_1} and \ref{mech_2} constitute perfect Bayesian equilibria with truthful reporting of traffic intensities when the future discount is sufficiently close to 1. Dynamic sharing schemes can achieve substantial improvement in spectral efficiency in comparison with static sharing schemes.

\section{Numerical Results} \label{Ch:unlicensed_static SE:Numerical Results}

    In this section, some numerical results are presented for the proposed schemes. Let ${r(\gamma)=\log_2(1+\gamma)}$, $\delta=0.99$, and $W=100$ MHz.

    First we show some numerical results for the entry game. Assume $\pi^i(x,\lambda)= \lambda r(P) x$, $P =100$ (i.e., 20 dB), and $\Exp \Lambda^i_t=\frac{1}{2}$ for every $t$ and $i$. Here the maximum number of active operators $n^*$ is the integer that satisfies
    \begin{align}
    50\log_2\left(1+\frac{100}{100(n^*-1)+1} \right) \geq c  \geq 50\log_2\left(1+\frac{100}{100n^*+1} \right).
    \end{align}
    Fig.~\ref{fig:number_cost} shows the relationship between $n^*$ and the investment cost ($c$). The trend in Fig.~\ref{fig:number_cost} appears to be approximately exponential.
    More operators are willing to invest as the cost goes down. Evidently, if the entry cost is 0, there will be infinitely many operators, who all receive zero revenue.

    \begin{figure}
    \centering
    \includegraphics[angle=360,width=.6\columnwidth]{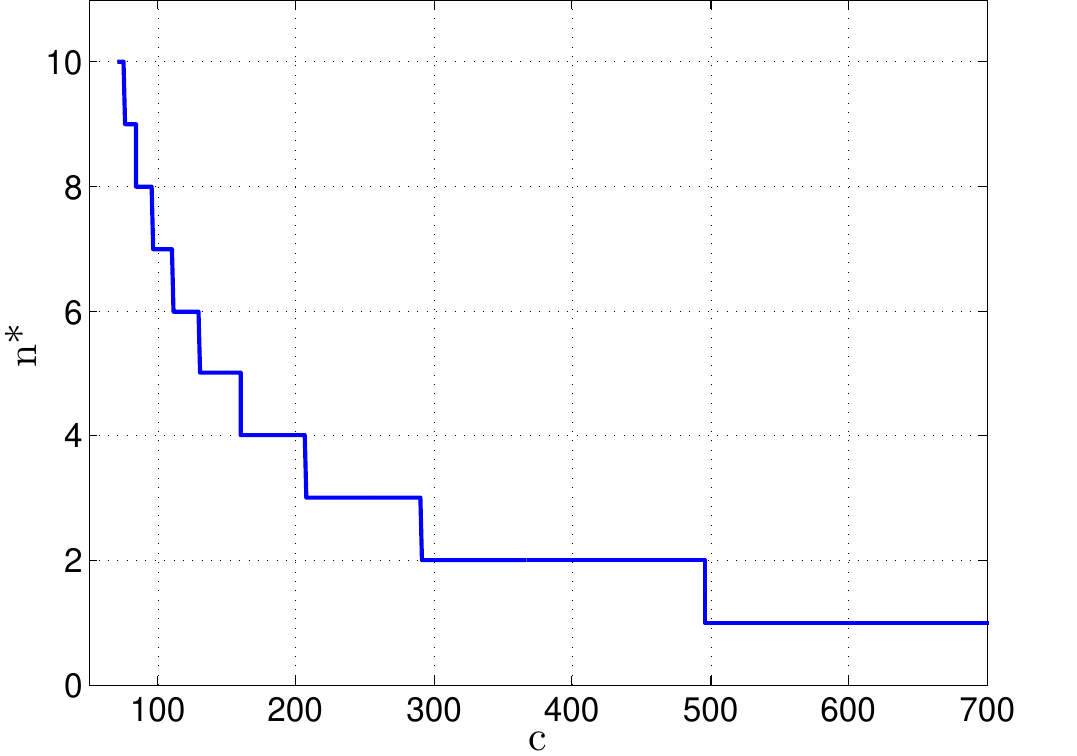}
    \caption{The maximum number of active operators for different investment costs in the entry game.}
    \label{fig:number_cost}
    \end{figure}

    \begin{figure}
    \centering
    \includegraphics[angle=360,width=.6\columnwidth]{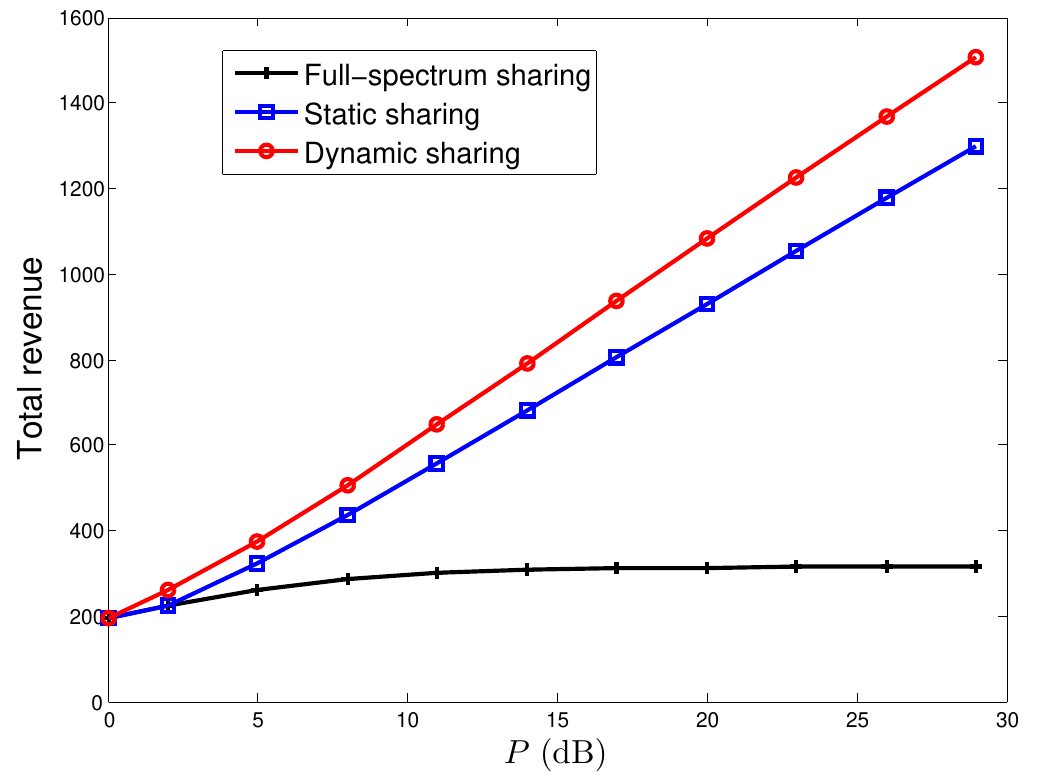}
    \caption{Total revenue for different sharing schemes: $\overline{b}=50$ MHz.}
    \label{fig:2operators_R}
    \end{figure}

    \begin{figure}
    \centering
    \includegraphics[angle=360,width=.6\columnwidth]{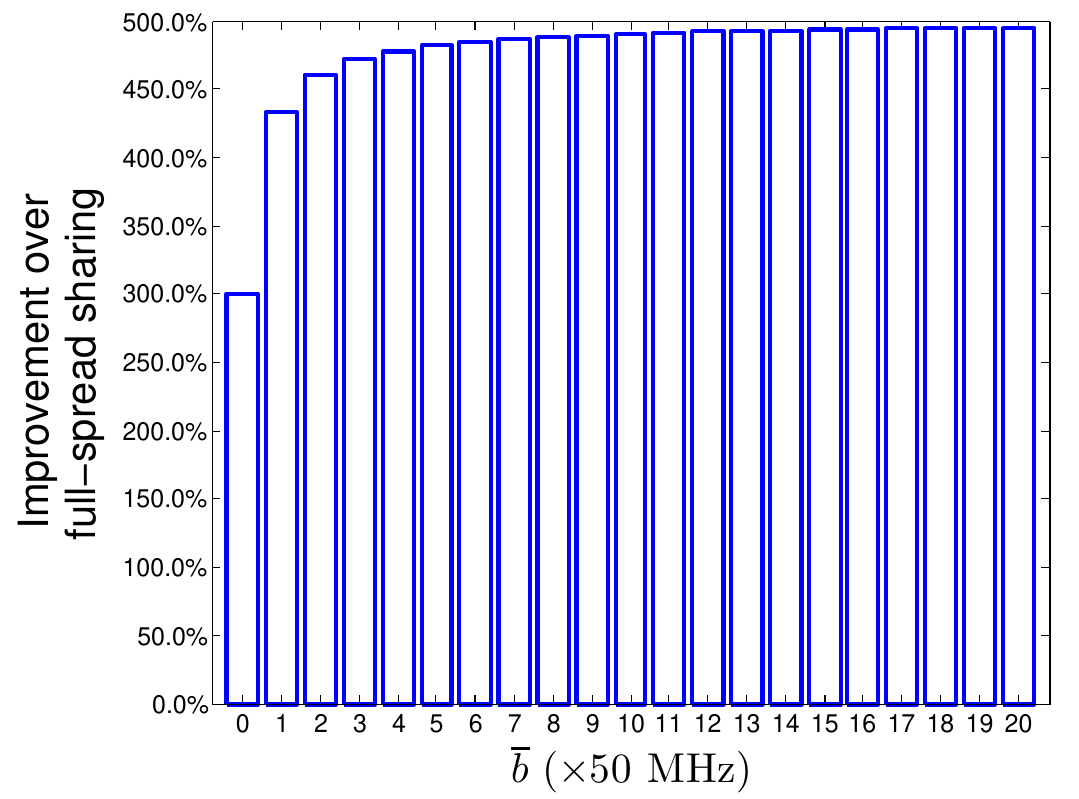}
    \caption{Improvement in total revenue by the proposed dynamic sharing for different balance limits: $\log_2(1+P)=8$.}
    \label{fig:2operators_B}
    \end{figure}

    Next, we present numerical results for the proposed static and dynamic schemes. We consider two operators and assume that the traffic intensity of an operator is either high $(\Lambda^i_t = 1)$ or low $(\Lambda^i_t = 0)$.
    Let $\pi^i(x,\lambda)= (24\lambda+1)^{0.5} (r(P)x)^{0.9}$, where $\pi^i(\cdot,\cdot)$ is a simple example of the Cobb-Douglas production function, which is widely used in economics to represent the technological relationship between the amount of output and the amounts of multiple inputs (traffic and spectrum here) \cite{Cobb-1928}.
    Assume that the traffic conditions of the two operators are independent and the probabilities of low traffic intensity for the two operators are $0.75$ and~$0.5$, respectively. The expected total revenues from full-spectrum sharing for the operators are $2\left(w\log_2\left(1+ \frac{P}{P + 1}\right) \right)^{0.9}$ and $3\left(w\log_2\left(1+ \frac{P}{P + 1}\right) \right)^{0.9}$; The expected total revenues in the proposed static scheme (Profile~\ref{mech_full_slot}) for the two operators are $2\left(w\log_2\left(1+ P\right) \right)^{0.9}$ and $3\left(w\log_2\left(1+ P\right) \right)^{0.9}$. In the proposed dynamic scheme (Profile~\ref{mech_1}), $\Delta$ is chosen to maximize the sum revenue of both operators and let Profile~\ref{mech_1} be a perfect Bayesian equilibrium.

    We compare the proposed static and dynamic schemes with full-spectrum sharing in Fig.~\ref{fig:2operators_R}. The proposed sharing schemes outperform full-spectrum sharing dramatically. As $\gamma_o$ goes up, the gain of the proposed sharing schemes increases. The proposed static scheme is better than full-spectrum sharing when $P >2.1$ dB, and the proposed dynamic scheme offers additional gain over the static sharing. When $P = 30$ dB, the proposed dynamic scheme has $400\%$ improvement over full-spectrum sharing and $16\%$ improvement over the proposed static scheme.
    Fig.~\ref{fig:2operators_B} shows the improvement from the proposed dynamic scheme over full-spread sharing under different balance limits. As $\overline{b}$ goes up, the gain approaches $500\%$ when $\log_2(1+P)=8$. Even if $\overline{b}$ is small, the proposed dynamic scheme still provides a significant improvement.

\section{Conclusion} \label{Conclusion}

    In this paper, we have studied unlicensed spectrum sharing by multiple strategic operators in a game theoretic framework.
    A static sharing scheme was first proposed for operators to share the spectrum in a given area and was shown to reach a subgame perfect Nash equilibrium. It has also been shown that the number of strategic operators willing to invest is limited due to entry barriers and externalities.
    A dynamic scheme for trading bandwidth has also been proposed, where operators with low traffic loads lend spectrum to those with high traffic loads, subject to a cumulative balance constraint on loaned bandwidth, which induces truthful reporting. Numerical results show that the proposed schemes can provide a substantial increase in spectral efficiency relative to full (uncoordinated) spectrum sharing. It is worth noting that a credible punishment scheme is in general necessary to deter deviation from the cooperation state.

 \section*{Acknowledgement}

    The authors thank the reviewers for their constructive comments. The authors also thank Dr. Weimin Xiao and Dr. Jialing Liu for useful discussions through the course of this project.

\appendices
\section{Proof of Theorem \ref{thm_dynamic}} \label{appandix:Thm2}

    According to \cite{Watson-2016}, Profile~\ref{mech_1} along with the system of beliefs $\pmb{\mu}$ is a perfect Bayesian equilibrium if Profile~\ref{mech_1} and $\pmb{\mu}$ satisfy the plain consistency conditions, the conditions that the operators' beliefs at the beginning of the game conform to the behavior strategy of Profile \ref{mech_1}, and that the strategy profile is sequentially rational given $\pmb{\mu}$.
    From the definition of $\pmb{\mu}$, the belief at the beginning of the game constructed based on Profile~\ref{mech_1}, so it conforms to the behavior strategy directly. According to operator $i$'s belief at its information set $h\in \mathcal{H}$, operator~$i$ believes that the actions at different information sets are independent. Also, from the definition of $\pmb{\mu}$, the strategy at any unobserved or off-path information set does not depend on the observations. Therefore, it is straightforward that the proposed system of beliefs $\pmb{\mu}$ satisfies the conditions of plain consistency in \cite[Definition~7]{Watson-2016}. Therefore, to prove this theorem, it is sufficient to show that there exist $\Delta>0$, and $T$ such that Profile~\ref{mech_1} is sequentially rational given $\pmb{\mu}$, as long as the conditions in the theorem hold.

    This notion of sequential rationality is defined in terms of what is commonly called ``one-shot deviations'', meaning that we evaluate operator $i$'s rationality at a given information set by looking at alternative choices only at this information set.
    In particular, the strategy profile is sequentially rational given $\pmb{\mu}$ if for every $h\in \mathcal{H}$, provided that information set $h$ is reached and operator $i$ moves at $h$, any one-shot deviation at $h$ results in no gain for operator $i$ given $\pmb{\mu}$.
    Because the total revenue is the sum of utilities in all slots, and the belief at each information set is concentrated on the strategy profiles described in Profile~\ref{mech_1}, it is equivalent to show that any one-shot deviation of Profile~\ref{mech_1} results in no gain if it happens on the path of the profile.
    The remainder of the proof consists of two parts, which address undetectable and detectable deviations, respectively.

    1) In this part, we discuss the one-shot deviation of reporting a false traffic intensity, which is not directly detectable. After the deviation, both operators conform to Profile~\ref{mech_1}.
    The subsequent actions of both operators are functions of the balance and the random traffic conditions.
    The key to the proof is to recognize that, because it is the most beneficial for an operator to borrow (resp., lend) when spectrum is most (resp., least) needed, an operator cannot gain by unilaterally lying about its own traffic. %This of course holds only if the discount factor is sufficiently close to 1.

    Without loss of generality, suppose operator 1 lies about its traffic intensity, whereas operator~2 reveals its true traffic intensity. Throughout the proof, we analyze the case where operator 1 deviates, so the superscript will be dropped when it is for operator 1 hereafter.
    The traffic loads $(\lambda_t,\lambda^2_t)$ have four possible realizations: (0,0), (0,1), (1,0), and (1,1). We discuss the case of $\lambda_t=\lambda^2_t=0$ in detail. The other three cases are similar and hence are omitted.
    If $b_t\geq -\overline{b}+\Delta$ and operator 1 lies to report high traffic, then operator 1 borrows spectrum to use bandwidth $w+\Delta$; otherwise, operator 1 uses  bandwidth $w$.
    Hence, when operator 1 lies to report high traffic in slot $t$, operator~1 may benefit by using $\Delta$ more bandwidth.
    The one-slot gain of operator 1 by lying is
    \begin{align} \label{eq_G02_1}
      \begin{split}
      G
      =\begin{cases}
        \pi(w+\Delta, 0)- \pi(w,0),  &\text{ if } b_t\geq -\overline{b}+\Delta \\ %DG: b\in(\overline{b}-\Delta, \overline{b}]
        0, &\text{ otherwise.}
              \end{cases}
      \end{split}
    \end{align}

    We develop a detailed proof with full justification (especially of the change of limits).
    If operator 1 does not use more spectrum from lying, i.e., $G = 0$, the deviating case is identical to the truth-telling case, so operator 1 does not gain by lying in this case. Hence we assume otherwise subsequently. If operator 1 uses more spectrum by lying, then due to the balance constraint, there will be one slot in the future, denoted as $t^*$, when operator 1 either borrows~$\Delta$ less (because it first hits the minimum balance) or lends $\Delta$ more (because it first hits the maximum balance) in comparison with the truth-telling case.
    The actions after slot $t^*$ are the same as in the truth-telling case.

    Let $p_{\tau} $ (resp., $q_{\tau}$) denote the probability that the first time operator 1 borrows less (resp., lends more) than the truth-telling case is $\tau$ slots after deviation. Clearly, ${\sum_{{\tau}=1}^\infty (p_{\tau} +q_{\tau} )=1}$. Note that $p_{\tau}$ and $q_{\tau}$ implicitly depend on the balance at the time of deviation.
    Since $\Prob(\Lambda_t=0,\Lambda_t^2=1)>0$ and $\Prob(\Lambda_t=1,\Lambda_t^2=0)>0$,
    we have $\sum_{\tau =1 }^{\infty}  p_{\tau} >0$ and $\sum_{\tau =1 }^{\infty} q_{\tau} >0$.
    Define the expected loss at the slot with $\tau$ slots after deviation as
    \begin{align} \label{eq:thm_2op_1}
    m_{\tau}=p_{\tau} (\pi(w+\Delta, 1) - \pi(w, 1))+q_{\tau} \left(\pi(w, 0)- \pi\left(w-\Delta, 0\right)\right).
    \end{align}
    By assumption, $\pi(\cdot,\cdot)$ is increasing in the first argument, so $m_{\tau}>0$. The total expected loss in the future is
    \begin{align} \label{eq:expected_loss}
      \begin{split}
      L  = \sum_{\tau =1 }^{\infty} \delta^{{\tau}} m_{\tau}.
      \end{split}
    \end{align}

    We show that $L > G$ if $\delta$ is sufficiently close to 1, i.e., the loss exceeds the one-slot gain. By~\eqref{property_b} and the continuity of $\pi$,
    \begin{align}
      \pi(w, 0) - \pi(w-\Delta, 0)
      < & \ \pi(w, 1) - \pi(w-\Delta, 1) \\
      = & \ \pi(w+\Delta, 1)- \pi(w, 1) + o(\Delta).
    \end{align}
    %where (3.20) holds as $\Delta\to 0$.
    Hence there must exist $\Delta>0$ which satisfies
    \begin{align}\label{ineq_1}
      \pi(w, 0)- \pi(w-\Delta, 0)< \pi(w+\Delta, 1) - \pi(w, 1).
    \end{align}
    Then, $\sum_{{\tau}=1}^{\infty}m_{\tau}$ is upper bounded, since from \eqref{eq:thm_2op_1} and \eqref{ineq_1}
    \begin{align}
      \sum_{{\tau}=1}^{\infty}m_{\tau}
      < \ & \sum_{{\tau}=1}^{\infty}  \left(p_{\tau} +q_{\tau} \right) (\pi(w+\Delta, 1) - \pi(w, 1)) \\
    = \ & \pi(w+\Delta, 1) - \pi(w, 1).
    \end{align}
    Since $m_{\tau}\geq 0$, $\sum_{{\tau}=1}^{l}m_{\tau}$ is nondecreasing in $l$.
    Due to monotonicity and boundedness, $\sum_{{\tau}=1}^l m_{\tau}$ converges as $l$ goes to infinity. Also, $|\delta^{{\tau}}m_{\tau}|\leq m_{\tau}$ since $\delta\in [0,1)$. According to~\cite[Theorem~7.10]{Rudin-1976}, $\sum_{{\tau}=1}^l \delta^{\tau} m_{\tau}$ converges uniformly.
    Also, by \cite[Theorem 7.11]{Rudin-1976}, we have
    \begin{align}
      \lim_{\delta \rightarrow 1}  \lim_{l \rightarrow \infty} \sum_{{\tau}=1}^l \delta^{\tau} m_{\tau}
    = & \lim_{l \rightarrow \infty}\lim_{\delta \rightarrow 1}   \sum_{{\tau}=1}^l \delta^{\tau} m_{\tau}   \\
    = & \lim_{l \rightarrow \infty}  \sum_{{\tau}=1}^l m_{\tau}.  \label{eq:thm_2op_4}
    \end{align}
    Thus, by \eqref{eq:thm_2op_1}, \eqref{eq:expected_loss}, \eqref{ineq_1} and \eqref{eq:thm_2op_4},
    \begin{align}
      \lim_{\delta \rightarrow 1} L
      \begin{split}
      = \ & \sum_{\tau =1 }^{\infty} m_{\tau}
      \end{split}  \label{ineq_16}\\
      > \ & \sum_{\tau =1 }^{\infty} \left(p_{\tau} +q_{\tau} \right)\left(\pi(w, 0)- \pi\left(w-\Delta, 0\right)\right) \\
      = \ &  \pi(w, 0)- \pi\left(w-\Delta, 0\right). \label{ineq_15}
    \end{align}
    Because $L $ increases with $\delta$, there exists $\delta_0$, such that for every $\delta >\delta_0$,
    \begin{align}
      \label{eq:4}
      L>\pi(w, 0)- \pi\left(w-\Delta, 0\right).
    \end{align}
    Due to~\eqref{property_a}, \eqref{eq_G02_1} and \eqref{eq:4}, we have $L\geq G$.
    Therefore, if both operators have low traffic intensities, then operator 1 suffers a net loss in the revenue by reporting high traffic intensity.
    Using similar arguments, one can show that an operator always suffers a net loss by unilaterally lying in the traffic under all traffic conditions.

    2) In this part, we consider detectable one-shot deviation, which will trigger the punishment state.  We shall show that when the duration of punishment is sufficiently long, it wipes out more than the gain from one-shot deviation.

    Denote the maximum utility of operator 1 with traffic $\lambda$ in one slot as $\overline{U}(\lambda)$ given by~\eqref{eq_upper_utility}.
    Let the expected total revenue of operator 1 starting from balance $b$ without deviation be denoted as $V(b)$.
    Because it is a repeated game where the only memory in the system is the balance, $V(b)$ does not depend on time.

    Consider a detectable one-shot deviation, where operator 1 deviates in the cooperation state in slot $t$ and then conforms thereafter. Operator 1 receives at most $\overline{U}(\lambda_t)$ in the slot when operator~1 deviates, experiences the punishment for $T$ slots and achieves $V(b_{t+1})$ thereafter. The total revenue of operator 1 from deviation is upper-bounded by
    \begin{align}
    V_{dev}(b_t) = \overline{U}(\lambda_t)+ \sum_{\tau=t+1}^{t+T} \delta^{\tau-t} \pi_f(\lambda_\tau) + \delta^{T+1} V(b_{t+1}).
  \end{align}
    If each operator conforms to the profile, operator 1 obtains
    \begin{align} \label{eq:11}
    V(b_t)=
    \pi(w_t,\lambda_t)+ \sum_{\tau=t+1}^{t+T} \delta^{\tau-t} \pi(w_\tau, \lambda_\tau) + \delta^{T+1} V(b_{t+T+1})
  \end{align}
  where $w_t$ is the bandwidth of operator 1 in slot $t$ when both operators conform.
    The gain by the deviation is less than
    \begin{align} \label{eq_gain_deviation}
    \begin{split}
      V_{dev}(b_t)-V(b_t)
    = \ & \overline{U}(\lambda_t)-\pi(w_t,\lambda_t) + \delta^{T+1} \left(V(b_{t+1})-V(b_{t+T+1})\right)  \\
    &- \sum_{\tau=t+1}^{t+T} \delta^{\tau-t} (\pi (w_\tau, \lambda_\tau) - \pi_f(\lambda_\tau)).
    \end{split}
    \end{align}

    We first show that $\overline{U}(\lambda_t)-\pi(w_t,\lambda_t)$ and $\delta^{T+1} \left(V(b_{t+1})-V(b_{t+T+1})\right)$ in \eqref{eq_gain_deviation} are upper bounded. Since the utility function is finite, $\overline{U}(\lambda_t)-\pi(w_t,\lambda_t)$ is upper bounded, assuming the upper-bound is $z_1$. If both operators conform, operator 1's utilities starting from balance $b$ and $b+\Delta$ are only different at one slot in the future, denoted as $t^*$.
    Thus, %from \eqref{eq:11},
    \begin{align} \label{eq:thm_2op_3}
      \begin{split}
       & V(b+\Delta)-V(b) \\
      = & \begin{cases}
        \delta^{t^*}(\pi(w+\Delta, 1)-\pi(w,1)), &\text{if it first hits the minimum balance}  \\
        \delta^{t^*}(\pi(w, 0)-\pi(w-\Delta,0)), &\text{if it first hits the maximum balance.}
        \end{cases}
      \end{split}
    \end{align}
    From \eqref{ineq_1}, \eqref{eq:thm_2op_3} and the fact that $\delta<1$,
    \begin{align} \label{eq:thm_2op_2}
    V(b+\Delta)-V(b) \leq \pi(w+\Delta, 1)-\pi(w,1).
    \end{align}
    From \eqref{eq:thm_2op_2} and the fact that $k= \frac{\overline{b}}{\Delta}$ and $\delta<1$,
     \begin{align}
     \delta^{T+1}\left(V(b_{t+1})-V(b_{t+T+1})\right) &\leq V(b_{t+1})-V(b_{t+T+1}) \\
     &\leq V(\overline{b})-V(-\overline{b}) \\
     & = \sum_{m=1}^{2k} \left(V(-\overline{b} + m\cdot \Delta)-V(-\overline{b}+(m-1)\cdot \Delta)\right) \\
     & \leq 2k(\pi(w+\Delta, 1)-\pi(w,1)).
     \end{align}
     Thus, $\delta^{T+1}\left(V(b_{t+1})-V(b_{t+T+1})\right)$ is upper bundered by $z_2=2k(\pi(w+\Delta, 1)-\pi(w,1))$.

     In the following, we show that $\sum_{\tau=t+1}^{t+T} \delta^{\tau-t} (\pi(w_\tau, \lambda_\tau)-\pi_f(\lambda_\tau))$ in \eqref{eq_gain_deviation} can be arbitrarily large for sufficiently large $T$ and $\delta$.
     By the assumption that $\pi(x,\lambda)$ is supermodular, for any $\xi>\lambda$,
    \begin{align}\label{eq:7}
    \pi(w,\xi) - \pi_f(\xi)> \pi(w,\lambda) - \pi_f(\lambda).
    \end{align}
    According to \eqref{property_a} and \eqref{eq:7}, we have
      \begin{align}
          \pi(w+\Delta, 1)-\pi_f(1)
       > & \ \pi(w, 1)-\pi_f(1) \\
       > & \ \pi(w, 0)-\pi_f(0)  \\
       > & \ \pi(w-\Delta, 0)-\pi_f(0). \label{ineq_14}
      \end{align}
      From Lemma~\ref{Orth_vs_Full}, there exist $\Delta>0$ and $z_3>0$ such that the right side of \eqref{ineq_14} is bigger than~$z_3$.
      Since $\pi(w_\tau, \lambda_\tau)$ takes value from $\{\pi(w+\Delta, 1), \pi(w, 1), \pi(w, 0), \pi(w-\Delta, 0)\}$ when both operators conform,
      \begin{align}
      \pi(w_\tau, \lambda_\tau) - \pi_f(\lambda_\tau)  \geq z_3.
      \end{align}
      There exists $T$ such that $ z_1 +z_2 - z_3T <0$,  i.e., the loss by punishment is larger than the gain by one-shot deviation.

      Therefore, if $T$ is large enough and $\delta$ is sufficiently close to 1, the detectable one-shot deviation in the cooperation state results in a lower revenue.

    If operator $1$ deviates in the punishment state in slot $t$, operator 1 obtains at most $\pi_f(\lambda_t)$ in the slot when operator 1 deviates, since $\pi_f(\lambda_t)$ is the min-max utility. Then the deviation only lengthens the punishment and postpones the larger utility in coordination state. Therefore, any detectable one-shot deviation is undesired.

    To summarize, if $\delta$ is sufficiently close to 1, any one-shot deviation by operator 1 is non-profitable.
    The same conclusion applies to operator 2 by symmetry.
    Therefore, there exist $\Delta>0$ and $T$, such that if $\delta$ is sufficiently close to 1, one-shot deviation on the path of the profile results in no gain. Profile~\ref{mech_1} is sequentially rational given $\pmb{\mu}$. Thus, Profile~\ref{mech_1} along with the system of beliefs $\pmb{\mu}$ is perfect Bayesian equilibrium.

\section{Proof of Theorem \ref{lemma_NOp_twoLev}} \label{appandix:Thm3}

    With the same arguments as in the proof of Theorem \ref{thm_dynamic}, the belief at the initial information set conforms to the behavior strategy of Profile~\ref{mech_2} for each operator and the proposed system of belief $\pmb{\mu}$ satisfies the conditions of plain consistency in \cite[Definition 7]{Watson-2016}.

    In order to prove the theorem, it is sufficient to show that there exist $\Delta>0$, and $T$ such that Profile~\ref{mech_2} is sequentially rational given $\pmb{\mu}$, as long as the conditions in Theorem~\ref{lemma_NOp_twoLev} hold. Similarly, it is equivalent to show that any one-shot deviation on the path of Profile~\ref{mech_2} results in no gain.
    Both undetectable and detectable deviations need to be addressed. The proof for the case of detectable deviation is essentially the same as that in Theorem~\ref{thm_dynamic}. Hence, here we only discuss the one-shot deviation of reporting a false traffic intensity, which is undetectable. %The proof is based on the one-shot deviation principle.

    For the $2$-operator case, an operator can access at most $\Delta$ more spectrum in one slot by lying in comparison with truthful reporting. However, for the $n$-operator case, it is possible for an operator to be chosen for trade no matter what the operator reports. If an operator is chosen to borrow (lend), the operator uses $\Delta$ Hz more (less) spectrum. Thus, an operator can access up to $2\Delta$ Hz more spectrum in one slot by lying.

    Without loss of generality, we analyze the case where operator 1 deviates, so the superscript will be dropped when it is for operator 1.
    Consider the event where the traffic intensity of operator $1$ is low. Denote the event that operator 1 would be chosen as a borrower if it lies to report a high traffic intensity by $F$. Denote the event that operator $1$ would be chosen as a lender if it tells the truth by $E$.
    When operator $1$ lies to report a high traffic intensity, operator $1$ may benefit by using more bandwidth in comparison with the truth-telling case. The gain by lying to report high traffic intensity is expressed as follows depending on the sub-events:
    \begin{align} \label{eq_G02_3}
    \begin{split}
      G
      =\begin{cases}
        \pi(w+\Delta, 0)- \pi(w,0), &\text{ if $F\cap \overline{ E}$} \\
        \pi(w, 0)- \pi\left(w- \Delta,0\right),      &\text{ if $\overline{F}\cap E$} \\
        \pi(w+\Delta, 0)- \pi(w-\Delta,0), &\text{ if $F\cap E$}\\  %DG: b\in(\overline{b}-\Delta, \overline{b}]
        0, &\text{ if $\overline{F} \cap  \overline{E}$.}
              \end{cases}
    \end{split}
    \end{align}

    All operators conform to Profile~\ref{mech_2} after the one-shot deviation.
    Similar as in the proof of Theorem~\ref{thm_dynamic}, due to the balance constraint, if operator $1$ uses more bandwidth by lying, operator~$1$ either borrows less (because it first hits the minimum balance) or lends more (because it first hits the maximum balance) in the future in comparison with the truth-telling case.

    It suffices to show that an operator cannot gain by lying about its own traffic for different cases in \eqref{eq_G02_3}.
    The proofs of the cases of $\overline{F}\cap E$ and $F\cap \overline{E}$ are similar to that of Theorem~\ref{thm_dynamic}. In the case of $\overline{F} \cap  \overline{E}$, there is no change in spectrum use and revenue by lying. If $F \cap E$ occurs, operator~$1$ uses $2\Delta$~Hz more spectrum when the deviation occurs. According to Profile~\ref{mech_2}, operator $1$ will use $\Delta$~Hz less spectrum in two slots in the future. Using the analogous argument as in \eqref{eq:4} in Theorem~\ref{thm_dynamic}, we can show that in each of the two slots when the loss occurs, the expected one-shot loss is greater than $\pi(w, 0)- \pi\left(w-\Delta, 0\right)$. Thus, the expected total loss of the one-shot deviation is greater than
    \begin{align}
        2(\pi(w, 0)- \pi\left(w-\Delta, 0\right)).
    \end{align}
    From \eqref{property_a},
    \begin{align}
    2(\pi(w, 0)- \pi\left(w-\Delta, 0\right)) &> \pi(w+\Delta, 0)-\pi(w, 0) + \pi(w, 0)- \pi\left(w-\Delta, 0\right)  \\
    &>\pi(w+\Delta, 0)- \pi\left(w-\Delta, 0\right).
    \end{align}
    That is, the expected total loss is greater than the one-shot gain when $F\cap E$ happens.
    Thus, if operator $1$ has a low traffic intensity, to report a high traffic intensity brings no gain.

    The case where the traffic intensity of operator $1$ is high can be treated similarly, and hence is omitted.
    Thus, if $\delta$ is sufficiently close to 1, any one-shot undetectable deviation by operator~$1$ is non-profitable.

    To summarize, one-shot deviation on the path of the profile in no gain. Profile~\ref{mech_2} is sequentially rational given $\pmb{\mu}$. Therefore, Profile~\ref{mech_2} along with the system of beliefs $\pmb{\mu}$ is a perfect Bayesian equilibrium.

\bibliographystyle{IEEEtran}

\end{document}